\documentclass{eptcs}
\providecommand{\event}{TFPIE 2020} % Name of the event you are submitting to
\usepackage{setspace}
\usepackage{textcomp}
\usepackage{amsmath}
\usepackage{amsthm}
\usepackage{amssymb}
\usepackage{alltt}
\usepackage{tikz}
\usetikzlibrary{arrows.meta}
\usetikzlibrary{automata,arrows}
\usepackage{subcaption}

\newcommand{\fsm}{\texttt{FSM}}
\newcommand{\racket}{\texttt{Racket}}
\newcommand{\dfa}{\texttt{DFA}}
\newcommand{\dfas}{\texttt{DFA}s}

\newtheorem{theorem}{Theorem}

\title{Visual Designing and Debugging of \\ Deterministic Finite-State Machines in \fsm}

\author{Marco T. Moraz\'an
\institute{Seton Hall University}
\email{morazanm@shu.edu}
\and
Joshua M. Schappel
\institute{Seton Hall University}
\email{schappjo@shu.edu}
\and
Sachin Mahashabde
\institute{Seton Hall University}
\email{mahashsa@shu.edu}
}
\def\titlerunning{FSM Visualization}
\def\authorrunning{Moraz\'an, Schappel, and Mahashabde}

\begin{document}
\maketitle

\begin{abstract}
This article presents a visualization tool for designing and debugging deterministic finite-state machines in \fsm--a domain specific language for the automata theory classroom. Like other automata visualization tools, users can edit machines and observe their execution, given some input. Unlike other automata visualization tools, the user is not burdened nor distracted with rendering a machine as a graph. Furthermore, emphasis is placed on the design of machines and this article presents a novel design recipe for deterministic finite-state machines. In support of the design process, the visualization tool allows for each state to be associated with an invariant predicate. During machine execution, the visualization tool indicates if the proposed invariant holds or does not hold after each transition. In this manner, students can validate and debug their machines before attempting to prove partial correctness or submitting for grading. In addition, any machine edited with the visualization tool can be rendered as executable code. The interface of the visualization tool along with extended examples of its use are presented.
\end{abstract}

\maketitle

\section{Introduction}
The typical Computer Science student, trained to program, finds an Automata Theory course very challenging and sometimes even overwhelming. This occurs because an Automata Theory course is typically taught in a manner that goes against the grain of what students have been taught. That is, students are asked to solve problems without being able to test and get immediate feedback on their solutions. For example, in a programming course, a student may be asked to write a program to decide if \texttt{s} is a substring of \texttt{S}. Typically, a student implements their solution in some programming language and gets immediate feedback from the compiler and unit tests. This immediate feedback is nonexistent when students are asked to develop a finite-state machine to solve the same problem by pencil-and-paper. This leads, more often than not, to buggy solutions and low grades.

The problem is further compounded if students \emph{design} their solutions. In programming, designing solutions means that loop \cite{Hoare} and accumulator \cite{HtDP,HtDP2} invariants are developed. Students may use Hoare logic to verify that invariants hold, thus giving students confidence in the algorithms they develop. No such framework is used in the typical Automata Theory classroom despite the fact that a finite-state machine is a representation of a program. Given that the machines developed by beginners may be unnecessarily complex (obfuscating their correctness), it is desirable to develop a teaching methodology and accompanying tools to assist students with the design of finite-state machines.

Given that it is reasonable for Computer Science students to be able to experiment with their designs, a domain-specific language (\texttt{DSL}), \fsm \ (\textbf{F}unctional \textbf{S}tate \textbf{M}achines), was developed \cite{fsm}. This \texttt{DSL} (embedded in \texttt{Racket}) allows students to implement finite-state machines. \fsm \ provides random testing facilities that provide students (and instructors) with immediate feedback on how well a given machine works. In addition, \fsm \ has a tailor-made error-messaging system that provides students with clear feedback messages \cite{fsmerrors}. This library has resulted in students being able to debug their machines before submitting them for grading. Furthermore, it allows students to implement the machine-building algorithms they develop as part of their constructive proofs. In this manner, students can test their algorithms before attempting to complete a formal proof. Thus, students experience less frustration and earn higher marks.

Although apathy towards Automata Theory has been reduced, many students still felt that they needed a tool to visualize their machines working. In fact, many students felt that they needed a tool to test state invariants during machine execution. This article presents a visualization tool to help students validate state invariants during the execution of a deterministic finite-state automaton (\dfa). Instead of focusing on state diagrams, as most visualization tools for finite-state machines do, the \fsm \ visualization tool focuses on state transitions and on state invariants. Students can provide input to a state machine and see the rules used to move from one state to another as input is consumed. In addition, students can provide invariant predicates for states. As machine execution advances, students can see if the invariants hold. The article is organized as follows. Section \ref{RW} discusses related work on \dfa \ design and visualization. Section \ref{fsm} briefly outlines the \fsm \ interface. Section \ref{design} outlines the \dfa \ design process taught to students. Section \ref{tool} describes the \fsm \ visualization tool's interface and how invariants are incorporated into the design process. Section \ref{ex} presents two examples of how the visualization tool is used by students. Finally, Section \ref{concl} presents concluding remarks and directions for future work.

\section{Related Work}
\label{RW}

\subsection{\dfa \ Design}
How to design \dfas \ is noticeably absent from most automata theory textbooks. The unspoken assumption is that students are simply able to design \dfas \ after studying a few examples. This is an unreasonable expectation just as is to expect students to be able to design programs after studying a few examples. The sole mention of design the authors are aware of is found in Rich's textbook \cite{Rich}. This textbook asks the reader to, in essence, think about the following two questions:
\begin{alltt}
     What does each of the machine's states need to record?

     If the language of the machine is infinite, what strings
     ought to take the machine to the same state?
\end{alltt}
The first question is asking for the role of each state to be defined. In other words, there must be an invariant property of the input read so far that is recorded in each state. The second question is asking to carefully design the transition function. No systematic steps are put forth to answer these questions. The work presented in this article also has students answer these questions. In contrast, however, it puts forth a systematic development strategy in the form of a design recipe \cite{HtDP,HtDP2}. To answer the first question, an informal statement describing the meaning of each state is required. In addition, these informal statements must be coded as invariant predicates used to validate the roles of the state. The second question is answered by developing a transition relation which is verifiable. That is, a proof of partial correctness must be developed. In essence, students develop a transition relation with an eye to developing a proof that establishes state invariant predicates always hold when a word is consumed by the machine.

The authors are not aware of any work that addresses the direct verification of \dfas \ in an automata theory course. Given that a \dfa \ models a program, it is natural to use proven techniques from program development in their construction. We build on the work done with design recipes \cite{HtDP, HtDP2,Mor7,Mor6,Mor5,Mor4,Mor3,Mor2,Mor1} to provide a framework students can follow to develop \dfas. Such a framework is important because it helps students to organize their thoughts and to communicate their solution to others--a fundamental pillar of programming \cite{PPL}. We also build on the work done in denotational semantics \cite{Hoare}. In programming, this is important because programmers that identify and write down invariants are more likely to write correct code \cite{Scott}. Our experience suggests that the same is true when developing \dfas.

\subsection{Visualization Tools}
The best-known visualization tool for finite-state machines is \texttt{JFLAP} \cite{Rodger}. It is used by several textbooks to aid in teaching the material covered in an Automata Theory or a Compilers class \cite{Gopal,Linz,Mozgovoy,Garcia}. \texttt{JFLAP} has a \texttt{GUI}-based interface that allows the user to graphically render \dfas \ as graphs. Like \fsm, users can apply a \dfa \ to an input string and observe how it moves from one state to the next as the input is consumed. Likewise, \fsm \ users can also see the machine's execution. Unlike \fsm, however, \texttt{JFLAP} requires that users graphically render the \dfa. This means that users are burdened with the actual drawing of the machine, which distracts them from the design of the machine. Another contrasting feature is that \texttt{JFLAP} does not provide the ability to associate an invariant with a state and observe if the invariant holds during execution. In \fsm, the user can validate their design using state invariants before attempting to verify their machines. In other words, \fsm's visualization tool allows users to debug their machines before attempting to prove the correctness of the machine. In addition, machines in \texttt{JFLAP} cannot be rendered as code in a higher-level programming language. In \fsm, any machine created or edited using the visualization tool can be rendered as \fsm \ code.

Another visualization tool for \dfas \ is \texttt{jFAST} \cite{White}. Similar to \texttt{JFLAP}, \texttt{jFAST} has users render graph-based representations of \dfas. Users can execute the machines by providing an input word. Unlike \texttt{JFLAP} and \fsm, \texttt{jFAST} does not provide an option to see the series of transitions made during machine execution nor a trace of said transitions. In addition, like \texttt{JFLAP}, users cannot associate invariants with states to assist them to validate their machines nor can they render their machines as programs in a higher-level programming language.

The \texttt{FSA} simulator also allows the user to work and experiment with \dfas \ \cite{Grinder}. Users can see a stepwise simulation of the execution of the machine on a given word. \texttt{FSA}, like \texttt{JFLAP} and \texttt{jFAST}, requires the user to create a graph-based representation of \dfas \ and does not have a mechanism to associate an invariant with a state. A distinguishing feature of \texttt{FSA} is a compare button that tests if two \dfas \ decide the same language. This button builds on work done on closure properties of regular languages. Specifically, two \dfas, \texttt{M1} and \texttt{M2}, decide the same language, \texttt{$L_M$}, if the language \texttt{$(\bar{L_{M1}} \cap L_{M2}) \cup (L_{M1} \cap \bar{L_{M2}})$} is empty, where $\bar{L_{X}}$ is the complement of the language of machine \texttt{X}. Buttons such as this compare button make little sense in an \fsm \ because students can easily define a function to test for language equality given two \dfas:
\begin{alltt}
     ; dfa dfa --> boolean
     ; Purpose: To determine if the two given DFAs accept the same language
     (define (same-language? M1 M2)
       (let ((M (sm-union (sm-intersection (sm-complement M1) M2)
                          (sm-intersection (sm-complement M2) M1))))
         (empty? (sm-getfinals M))))
\end{alltt}
This function is, in essence, a direct translation of the equation above that uses the \fsm \ primitives \texttt{sm-union}, \texttt{sm-intersection}, \texttt{sm-complement}, and \texttt{sm-getfinals} which, respectively, take the union of two state machines, take the intersection of two state machines, take the complement of a state machine, and access the set of final states of a state machine. We argue that providing the ability to easily write such functions is more conducive to learning than having a button in a visualization tool to perform such tasks.

\section{\fsm \ Interface}
\label{fsm}

A \dfa \ is typically defined as a quintuple $M = (K, \Sigma, s, F, \delta)$, where $K$ is a finite set of states, $\Sigma$ is an alphabet, $s \in K$ is the staring state, $F \subseteq K$ is the set of final states, and $\delta$, the \emph{transition function}, is a function $K \ \Sigma \rightarrow K$. The input to $M$ is a word $w \in \Sigma^*$. The entire word is consumed, one symbol at a time, from left to right. The transition function specifies how the machine picks its next state after consuming $a \in \Sigma$ from the input word. After consuming the input word, the machine accepts if it has reached a final state. Otherwise, the machine rejects.

\fsm \ uses the following definitions:
\begin{itemize}
  \item State = symbol
  \item Alphabet = (listof symbol)
  \item Word = (listof symbol)
  \item Transition = (listof State symbol State)
  \item Transitions = (listof Transition)
  \item Result = \{accept, reject\}
  \item Configuration = (list Word State)
  \item Trace = (append (listof Configuration) (list Result))
\end{itemize}
In \fsm, a \dfa \ is an interface for the following services:
\begin{itemize}
  \item \textbf{make-dfa}: (listof State) Alphabet State (listof State) Transitions [\textquotesingle{no-dead}] $\rightarrow$ dfa\\
      Purpose: To construct a \dfa. The optional symbol indicates not to add a dead state
  \item \textbf{sm-getstates}: dfa $\rightarrow$ (listof State) \\
        Purpose: To access the given \dfa's set of states
  \item \textbf{sm-getalphabet}: dfa $\rightarrow$ Alphabet \\
        Purpose: To access the given \dfa's alphabet
  \item \textbf{sm-getstart}: dfa $\rightarrow$ State \\
        Purpose: To access the given \dfa's starting state
  \item \textbf{sm-getfinals}: dfa $\rightarrow$ (listof State) \\
        Purpose: To access the given \dfa's set of final states
  \item \textbf{sm-getrules}: dfa $\rightarrow$ Transitions \\
        Purpose: To access the given \dfa's transition function
  \item \textbf{sm-apply}: dfa Word $\rightarrow$ Result\\
        Purpose: To apply the given \dfa \ to the given word
  \item \textbf{sm-showtransitions}: dfa Word $\rightarrow$ Trace\\
        Purpose: To return the trace of applying the given \dfa \ to the given word
  \item \textbf{sm-test}: dfa [natnum] $\rightarrow$ (listof (list Word Result))\\
        Purpose: To return the results obtained from applying the given machine to the given optional number of randomly generated input words. If the number of tests is not provided, the default is 100 tests.
\end{itemize}
If the constructor, \texttt{make-dfa}, is given improper input, the tailor-made \fsm \ error-messaging system throws an error \cite{fsmerrors}. If the provided transitions are nondeterministic, for example, a error with a descriptive message containing the violating transitions is generated.

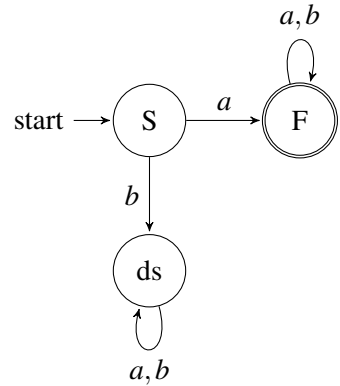
\begin{figure}[t]
\begin{center}
\begin{tikzpicture}[>=stealth',shorten >=1pt,auto,node distance=2 cm, scale = 1, transform shape]
%\begin{scope}[every node/.style={circle,thick,draw}]
    \node[initial,state] (S) {S};
    \node[state,accepting] (F) [right of=S] {F};
    \node[state] (ds) [below of=S] {ds};
%\end{scope}
%\begin{scope}[>={Stealth[black]}, every node/.style={fill=white,circle}, every edge/.style={draw=black,thick}]
    \path [->] (S) edge[align=center] node {$a$} (F);
    \path [->] (S) edge[align=center] node[left] {$b$} (ds);
    \path [->] (F) edge[loop above,align=center] node {$a,b$} (F);
    \path [->] (ds) edge[loop below,align=center] node {$a,b$} (ds);
%\end{scope}
\end{tikzpicture}
\end{center}
\caption{\dfa \ to decide the language of strings that start with an \texttt{a}.}
\label{a*}
\end{figure}

To make the use of \fsm \ concrete, consider deciding the following regular language:
\begin{alltt}
     L = \{w | w \(\in\) \((a, b)\sp{*}\) \(\wedge\) w starts with an a\}.
\end{alltt}
A graphical depiction of the automaton to decide \texttt{L} is displayed in Figure \ref{a*}. Students can implement this automaton in \fsm \ as follows:
\begin{alltt}
     (define a* (make-dfa \textquotesingle(S F)       ; the states
                          \textquotesingle(a b)       ; the input alphabet
                          \textquotesingle{S}           ; the starting state
                          \textquotesingle(F)         ; the set of final states
                          \textquotesingle((S a F)    ; the transition function
                            (F a F)
                            (F b F))))
\end{alltt}
For \dfas, students do not need to explicitly include a dead state (i.e., \texttt{ds}). \fsm \ automatically adds a dead state and the corresponding transitions into the dead state for any missing transitions (unless explicitly indicated by providing the optional \textquotesingle{\texttt{no-dead}} argument).

Now, instead of simply having an abstract drawing, students have an executable machine which they can apply to a word. For example, students may apply \texttt{a*} to \textquotesingle{\texttt{(b a a b)}}:
\begin{alltt}
     > (sm-apply a* \textquotesingle{(b a a b)})
     \textquotesingle{reject}
\end{alltt}
They can also see the different configurations of \texttt{a*} as it consumes an input word. For example, the following are the configurations of \texttt{a*} as it consumes \textquotesingle{\texttt{(a b a b a)}}:
\begin{alltt}
     > (sm-showtransitions a* \textquotesingle{(a b a b a)})
     \textquotesingle{(((a b a b a) S)
       ((b a b a) F)
       ((a b a) F)
       ((b a) F)
       ((a) F)
       (() F)
       accept)}
\end{alltt}

Finally, students can test a machine with randomly generated words to help validate their design. The following is the result of testing \texttt{a*} with 10 randomly generated words:
\begin{alltt}
     > (sm-test a* 10)
     \textquotesingle(((a b b b a) accept)
       ((a a a a b b a b) accept)
       ((a b) accept)
       ((b b a) reject)
       ((a a) accept)
       ((a a b) accept)
       ((b) reject)
       (() reject)
       ((b a a a b a a) reject)
       ((b b a a b b a b b) reject))
\end{alltt}

The above capabilities give students the opportunity to get immediate feedback and fine-tune their designs. The feedback comes from two sources: testing and a tailored-made error-messaging system. This is important because students commonly have buggy designs that lead to lower grades. When using \fsm, it is not uncommon for students to detect bugs by simply being able to test their machines as illustrated above. This ability to validate their machines gives students a degree of confidence in their designs. It cannot, however, give them full confidence as it falls short of verifying their machines. It is also not uncommon for students to discover that, for example, their transition function refers to non-existent states. This feedback is obtained from the \fsm \ error-messaging system. For a more thorough description of \fsm \ \cite{fsm} and of the \fsm \ error-messaging system~\cite{fsmerrors}, the reader is referred to the cited publications.

\section{Designing Deterministic Finite-State Automata}
\label{design}
Designing \dfas \ is no different than designing programs--a point commonly missed in Automata Theory textbooks. Students must analyze a problem, identify conditions that must hold as input is processed, associate these invariant conditions with a state, and translate their design into a function. To foster a design methodology, students are presented with the following design recipe to develop a deterministic finite state machine, \texttt{M}, for some language \texttt{L}:
\begin{quote}
\begin{enumerate}
  \item Name the \dfa \ and specify the input alphabet.
  \item Write unit tests.
  \item Determine the conditions that must be tracked as input is consumed and associate each condition with a state clearly identifying the start state and the final states.
  \item Formulate an invariant predicate for each state.
  \item Formulate the transition function.
  \item Implement the machine.
  \item Test the machine using unit tests and random testing.
  \item Develop a proof of partial correctness.
  \item Prove \texttt{L} = $\mathfrak{L}$(\texttt{M}) (the language decided by \texttt{M}).
\end{enumerate}
\end{quote}

In Step 1, students choose a name for their \dfa \ and specify the alphabet for input strings. The chosen name and alphabet are used to write unit tests in Step 2. The tests for Step 2 must include examples of words that ought to be accepted and ought to be rejected. If the set of accepted strings or the set of rejected strings is empty, then no unit tests are written for said set.

In Step 3, students begin to design their algorithm. Students must analyze the problem to determine what conditions must be tracked as the input is consumed to conclude that it is accepted or rejected. Each identified condition is associated with a state and represents an invariant for the state. In addition, students identify the starting state and the accepting/final states.

In Step 4, students formalize the meaning of each state based on the conditions identified in Step 3. For each state, a student must write a predicate for the invariant condition that a state represents. These predicates have one parameter for the consumed part of the input word \texttt{w}. To guide the development of the predicates, students are told that the invariant for a final state must imply that \texttt{w} $\in$ \texttt{L} and that the invariant of a non-final state must imply \texttt{w} $\notin$ \texttt{L}. In addition, students are told that for any rule, \texttt{(Q a P)}, \texttt{Q}'s invariant and the consumption of \texttt{a} must imply \texttt{P}'s invariant.

In Step 5, students formulate the transition function. Given that a \dfa \ is being designed, for each state, there must be one transition for each element of the alphabet. The choice of destination state for each rule is made based on what invariant condition holds after the next character in the input string is consumed in the current state.

In Step 6, students implement their machines in \fsm. This means providing the constructor for \dfas \, \texttt{make-dfa}, the set of states and the alphabet from Step 1, the starting and final states from Step 3, and the transition function from Step 5.

In Step 7, students test their machines. They must make sure that their unit tests are thorough and pass. In addition, students must use random testing to further validate their machines following current trends to generate tests instead of (just) writing tests \cite{Hughes}. Random testing is a rather important step as it gives students confidence in their design.

Finally, in Steps 8 and 9 students develop proofs. For step 8, students establish the partial correctness of the machine's transition function by proving that the state invariants from Step 4 always hold. They achieve this by induction on the length of, \texttt{w}, the input word. Step 9 requires proving that \texttt{L} = $\mathfrak{L}$(\texttt{M}). That is, students prove that \texttt{L} is the language decided by \texttt{M}. They achieve this by proving following equivalences:
  \begin{enumerate}
    \item \texttt{w} $\in$ $\mathfrak{L}$(\texttt{M}) $\Leftrightarrow$ \texttt{w} $\in$ \texttt{L}\\
    \item \texttt{w} $\notin$ $\mathfrak{L}$(\texttt{M}) $\Leftrightarrow$ \texttt{w} $\notin$ \texttt{L}
  \end{enumerate}
Students exploit the proof that the invariants always hold to complete this step.

\begin{figure}[t]
\begin{center}
\begin{tikzpicture}[>=stealth',shorten >=1pt,auto,node distance=2 cm, scale = 1, transform shape]
%\begin{scope}[every node/.style={circle,thick,draw}]
    \node[initial,state] (S) {S};
    \node[state,accepting] (F) [right of=S] {F};
    \node[state] (A) [right of=F] {A};
    \node[state] (ds) [below of=F] {ds};
%\end{scope}
%\begin{scope}[>={Stealth[black]}, every node/.style={fill=white,circle}, every edge/.style={draw=black,thick}]
    \path [->] (S) edge[align=center] node {$a$} (F);
    \path [->] (S) edge[align=center] node[below] {$b$} (ds);
    \path [->] (F) edge[align=center,bend left=10] node {$b$} (A);
    \path [->] (F) edge[loop above,align=center] node {$a$} (F);
    \path [->] (A) edge[loop above,align=center] node {$b$} (A);
    \path [->] (A) edge[align=center,bend left=10] node {$a$} (F);
    \path [->] (ds) edge[loop below,align=center] node {$a,b$} (ds);
%\end{scope}
\end{tikzpicture}
\end{center}
\caption{\dfa \ to decide the language of strings that start and end with an \texttt{a}.}
\label{a*a}
\end{figure}
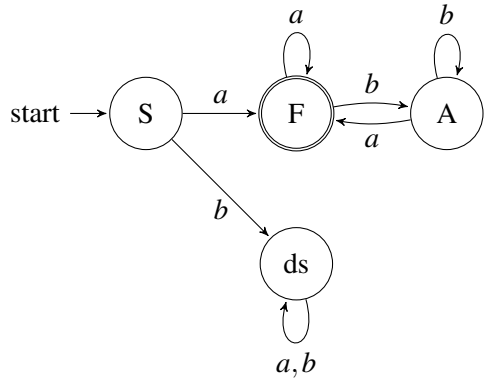

To make the steps of the design recipe concrete, consider designing a \dfa \ for:
\begin{alltt}
     L = \{w | w \(\in (a, b)\sp{*} \wedge \) w starts and ends with an a\}
\end{alltt}

To satisfy Step 1, we name the \dfa \ \texttt{a*a} and identify the alphabet as \texttt{\textquotesingle (a b)}. Step 2 results, for example, in the following unit tests using \texttt{Racket}'s testing engine:
\begin{alltt}
     (check-expect (sm-apply a*a \textquotesingle()) \textquotesingle{reject})
     (check-expect (sm-apply a*a \textquotesingle(b a a)) \textquotesingle{reject})
     (check-expect (sm-apply a*a \textquotesingle(a b a b)) \textquotesingle{reject})
     (check-expect (sm-apply a*a \textquotesingle(a)) \textquotesingle{accept})
     (check-expect (sm-apply a*a \textquotesingle(a b a a)) \textquotesingle{accept})
     (check-expect (sm-apply a*a \textquotesingle(a a a)) \textquotesingle{accept})
\end{alltt}
Observe that the suite of unit tests includes examples of both accepted and rejected strings. In addition, tests are included for both varieties of Word (i.e., empty and non-empty). Given that students have previously completed a year-long course using \textit{How to Design Programs} \cite{HtDP,HtDP2} at Seton Hall University, it is common for students to be thorough about covering the different varieties of the input. Nonetheless, there are always students that require reminding and ecouragement.

Step 3 reveals that the following four conditions must be tracked using 4 different states as the input word is consumed:
\begin{quote}
\begin{description}
  \item[S] The starting state represents that no input has been consumed.
  \item[F] This single final state represents that the input consumed starts and ends with an \texttt{a}.
  \item[A] Represents that the consumed input starts with an \texttt{a} and does not end with an \texttt{a}.
  \item[ds] The dead state represents that the input consumed does not start with an \texttt{a}.
\end{description}
\end{quote}

Step 4 yields four predicates. Each predicate represents the invariant condition that must hold for the consumed input. Concretely, the above invariant conditions may be, for example, implemented as displayed in Figure \ref{a*ainvs}. It is noteworthy that the instructor must remind students that they must write unit tests for invariant predicates.

\begin{figure}[t]
\begin{alltt}
     (define S-INV null?)

     (check-expect (S-INV \textquotesingle{(})) #t)  (check-expect (S-INV \textquotesingle{(}a b a)) #f)

     (define (F-INV consumed-input)
       (and (not (null? consumed-input))
            (eq? (first consumed-input) \textquotesingle{a})
            (eq? (last consumed-input) \textquotesingle{a})))

     (check-expect (F-INV \textquotesingle{(})) #f)   (check-expect (F-INV \textquotesingle{(}b b b)) #f)
     (check-expect (F-INV \textquotesingle{(}a)) #t)  (check-expect (F-INV \textquotesingle{(}a b a)) #t)

     (define (A-INV consumed-input)
       (and (not (null? consumed-input))
            (eq? (first consumed-input) \textquotesingle{a})
            (not (eq? (last consumed-input) \textquotesingle{a}))))

     (check-expect (A-INV \textquotesingle{(})) #f)  (check-expect (A-INV \textquotesingle{(}b b)) #f)
     (check-expect (A-INV \textquotesingle{(}a)) #f) (check-expect (A-INV \textquotesingle{(}a b b)) #t)

     (define (DS-INV consumed-input)
       (and (not (null? consumed-input))
            (not (eq? (first consumed-input) \textquotesingle{a}))))

     (check-expect (DS-INV \textquotesingle{(})) #f)  (check-expect (DS-INV \textquotesingle{(}a a b b)) #f)
     (check-expect (DS-INV \textquotesingle{(}b)) #t) (check-expect (DS-INV \textquotesingle{(}b a a a b)) #t)
\end{alltt}
\caption{Invariant Predicates for \texttt{a*a}.}
\label{a*ainvs}
\end{figure}

For Step 5, students must develop the transition function. This means, that for every state, there must be an outgoing transition for every element of the input alphabet. Students are taught to develop the transition function using the state invariants. Observe, for example, that if we assume \texttt{S-INV} holds, then consuming an \texttt{a} implies \texttt{F-INV} holds and consuming a \texttt{b} implies \texttt{DEAD-INV} holds. This tells us that \texttt{(S a F)} and \texttt{(S a ds)} may be used as machine transitions. Similar analysis for the other states leads to the complete set of transitions (using \fsm \ syntax): \texttt{\textquotesingle{((S a F) (F a F) (F b A) (A a F) (A b A))}}. Consistent with the use of \fsm, students are not required to list transitions involving the dead state, \texttt{ds}, but are free to do so. Students that are more visually oriented prefer to draw the \dfa, displayed in Figure \ref{a*a}, to illustrate the transition function. Both strategies for satisfying Step 5 are acceptable. Students that prefer to display the image of the automaton are free to not include transitions involving the dead state as it is understood that all missing transitions are to the dead state. Most students, however, choose to include the dead state in their diagrams. It is worth noting that, by using the state invariants to develop the transition function, we teach students that a proof, just like a program, is designed. This idea also fits in nicely with the development of constructive proofs.

\begin{figure}[t]
\begin{theorem}
For (\texttt{S}, \texttt{w}) $\Rightarrow^{*}_{\texttt{a*a}}$ (\texttt{H}, $\epsilon$) invariants \ hold, where \texttt{H} $\in$ \texttt{K}$_{\texttt{a*a}}$ and \texttt{w} \(\in\) \texttt{\(\Sigma\sp{*}\sb{a*a}\)}
\label{thm1}
\end{theorem}
\begin{proof}
\begin{description}
  \item[Base case]
    \[
    \begin{array}{rcl}
     |\texttt{w}| = 0   & \Rightarrow & \textsc{ci} = \epsilon \Rightarrow \texttt{ S-INV} \\
    \end{array}
    \]
  \item[Inductive Step]
    \[
    \begin{array}{l}
      \textrm{Assume invariants hold for } |\texttt{w}| = k. \textrm{     Show that invariants hold for } |\texttt{w}| = k+1.\\
      |\texttt{w}| = k+1 \Rightarrow \texttt{w} = \texttt{vi}, \textrm{ where } \texttt{v} \in \texttt{$\Sigma^*_{a*a}$} \textrm{ and } \texttt{i} \in \texttt{$\Sigma^*_{a*a}$} \\
      \textrm{For } \texttt{(S, vi)} \Rightarrow^{*}_{a*a}  \texttt{(M, i)} \textrm{ invariants hold by inductive hypothesis} \\
      \textrm{We must show that for } \texttt{(M i Q)} \texttt{ Q-INV} \textrm{ holds}.
    \end{array}
    \]
    \[
    \begin{array}{rll}
      \texttt{S-INV} & \Rightarrow & \textsc{ci} = \epsilon\\
      \textrm{Using } (\texttt{S } \texttt{a } \texttt{F}) & \Rightarrow & \textsc{ci} = \texttt{a} \Rightarrow \texttt{ F-INV}\\
      \textrm{Using } (\texttt{S } \texttt{b } \texttt{ds}) & \Rightarrow & \textsc{ci} = \texttt{b} \Rightarrow \texttt{ DS-INV}\\
      & &\\
      \texttt{F-INV} & \Rightarrow & \textsc{ci} = \texttt{a } \vee \texttt{ a(a|b)}^*\texttt{a}\\
      \textrm{Using } (\texttt{F } \texttt{a } \texttt{F}) & \Rightarrow & \textsc{ci} = \texttt{a(a|b)}^*\texttt{a } \Rightarrow \texttt{ F-INV}\\
      \textrm{Using } (\texttt{F } \texttt{b } \texttt{A}) & \Rightarrow & \textsc{ci} = \texttt{a(a|b)}^*\texttt{b } \Rightarrow \texttt{ A-INV}\\
      & &\\
      \texttt{A-INV} & \Rightarrow & \textsc{ci} = \texttt{a(a|b)}^*\texttt{b }\\
      \textrm{Using } (\texttt{A } \texttt{a } \texttt{F}) & \Rightarrow & \textsc{ci} = \texttt{a(a|b)}^*\texttt{a } \Rightarrow \texttt{ F-INV}\\
      \textrm{Using } (\texttt{A } \texttt{b } \texttt{A}) & \Rightarrow & \textsc{ci} = \texttt{a(a|b)}^*\texttt{b } \Rightarrow \texttt{ A-INV}\\
      & &\\
      \texttt{DS-INV} & \Rightarrow & \textsc{ci} = \texttt{b(a|b)}^*\\
      \textrm{Using } (\texttt{ds } \texttt{a } \texttt{ds}) & \Rightarrow & \textsc{ci} = \texttt{b(a|b)}^* \Rightarrow \texttt{ DS-INV}\\
      \textrm{Using } (\texttt{ds } \texttt{b } \texttt{ds}) & \Rightarrow & \textsc{ci} = \texttt{b(a|b)}^* \Rightarrow \texttt{ DS-INV}
    \end{array}
    \]
\end{description}
\end{proof}
\caption{Proof of partial correctness for \texttt{a*a}.}
\label{inv1}
\end{figure}

For Step 6, students implement their design as an \fsm \ program. After following the first 5 steps of the design recipe, students find code development mostly straightforward. The following is a sample result for this step:
\begin{alltt}
     (define a*a (make-dfa \textquotesingle(S F A) \textquotesingle(a b) \textquotesingle{S} \textquotesingle(F)
                           \textquotesingle((S a F)
                             (F a F)
                             (F b A)
                             (A a F)
                             (A b A))))
\end{alltt}

For Step 7, students must run their unit tests. If any tests fail, they must first make sure their tests are correctly written. In addition to running unit tests, students must use random testing (using \texttt{sm-test}). If any unit tests fail or any random test produces the wrong result, students need to redesign their machine.

For Step 8, students develop an inductive proof to establish that \texttt{a*a}'s invariants always hold when consuming an arbitrary \texttt{w}. The proof for \texttt{a*a} is displayed in Figure \ref{inv1}, where \texttt{K}$_{\texttt{a*a}}$ and \texttt{$\Sigma$}$_{\texttt{a*a}}$, respectively, denote the states and the alphabet of \texttt{a*a}. The symbol $\Rightarrow_{\texttt{a*a}}$ denotes one step of the machine and $\Rightarrow^{*}_{\texttt{a*a}}$ denotes zero or more steps of the machine. To denote the consumed input \textsc{ci} is used. If the students have correctly developed the invariants for Step 4 and have used them to design the transition function for Step 5, they find developing the argument for partial correctness straightforward. This step, however, is challenging for many students. Here is where they discover that mistakes in Steps 4 and 5 must be corrected.

Once students establish that the invariants always hold, they can proceed to establish that the machine decides the desired language to satisfy Step 9. For our example, it is important to observe that only \texttt{F-INV} holds when the input word starts and ends with an \texttt{a}. Using this observation and the proof of partial correctness, the students must prove that \texttt{L = $\mathfrak{L}$(a*a)} by proving two lemmas. The first is:\\ \\
\begin{alltt}
     w \(\in\) L \(\Leftrightarrow\) w \(\in\) \(\mathfrak{L}\)(a*a)
\end{alltt}
We outline the proof as follows:\\
\texttt{($\Rightarrow$) Assume: \texttt{w} $\in$ L \ \ \ \ \ Show: \texttt{w} $\in$ $\mathfrak{L}$(\texttt{a*a})}
\[\begin{array}{lcl}
  \texttt{w} \in \texttt{L} & \Rightarrow & \texttt{w} = \texttt{a } \vee \texttt{ w} = \texttt{a}y\texttt{a}, \ \text{where } y \in \texttt{$\Sigma$}^*_{\texttt{a*a}}\\
          & \Rightarrow & \texttt{a*a} \ \text{halts in \texttt{F}, given that only \texttt{F-INV} can hold after consuming \texttt{w}}\\
          & \Rightarrow & \texttt{w} \in \mathfrak{L}(\texttt{a*a})\\
\end{array}\]

\noindent \texttt{($\Leftarrow$) Assume: w $\in$ $\mathfrak{L}$(a*a) \ \ \ \ \ Show: w $\in$ L}
\begin{center}
$\begin{array}{lcl}
  \texttt{w} \in \mathfrak{L}(a*a) & \Rightarrow &  \texttt{a*a} \ \text{halts in} \ \texttt{F} \\
                                   & \Rightarrow & \texttt{w} \ \text{starts and ends with an} \ \texttt{a} \textrm{, given } \texttt{F-INV}  \text{ holds}\\
                                   & \Rightarrow & \texttt{w} \in \texttt{L}\\
\end{array}$
\end{center}

\noindent The second is:
\begin{alltt}
     w \(\notin\) L \(\Leftrightarrow\) w \(\notin\) \(\mathfrak{L}\)(a*a)
\end{alltt}
We outline the proof as follows: \\
\texttt{($\Rightarrow$) Assume: w $\notin$ L \ \ \ \ \ Show: w $\notin$ $\mathfrak{L}$(a*a)}
\begin{center}
$\begin{array}{lcl}
  w \notin L & \Rightarrow & \texttt{w} = \texttt{b}y \ \vee \texttt{ w} = y\texttt{b}, \ \text{where \ } y \in \texttt{$\Sigma$}^*_{\texttt{a*a}}\\
             & \Rightarrow & \texttt{a*a} \ \text{does not halt in} \ \texttt{F} \ \text{after consuming \texttt{w}, given} \ \texttt{F-INV} \ \text{does not hold}\\
             & \Rightarrow & \texttt{w} \notin \mathfrak{L}(\texttt{a*a})\\
\end{array}$
\end{center}

\noindent \texttt{($\Leftarrow$) Assume: w $\notin$ $\mathfrak{L}$(a*a) \ \ \ \ \ Show: w $\notin$ L}
\begin{center}
$\begin{array}{lcl}
  \texttt{w} \notin \mathfrak{L}(a*a) & \Rightarrow & \texttt{a*a} \ \text{does not halt in} \ \texttt{F}\\
                                      & \Rightarrow & \texttt{w} \ \text{does not start and end with an} \ \texttt{a}, \ \text{given} \ \texttt{F-INV} \ \text{does not hold}\\
                                      & \Rightarrow & \texttt{w} \notin \texttt{L}
\end{array}$
\end{center}

\section{The Visualization Tool}
\label{tool}

It is not surprising that a significant number students do not correctly design a machine at first and must engage in a process of iterative refinement to eliminate bugs. The steps of the design recipe that are most challenging for students are invariant predicate development and the proof of partial correctness. More specifically, students find it difficult to manually verify if a state's invariant holds. This task, admittedly, is trivial for toy examples, but with only a modicum increase in complexity, manually analyzing the consumed input by hand is difficult. The process is cumbersome and repetitive as verifying invariants must be done for every single transition. This type of repetitive chore is prone to errors as a student's attention span is likely to slip. The reader can imagine that manual invariant verification becomes harder as the input becomes longer.

Students suggested that they need a way to visualize the machine's execution. They pointed out that in a regular programming language, they could always, for example, populate their programs with print statements to examine values and to test conditions. This is a very reasonable point and became the root motivation to develop a visualization tool for deterministic finite-state automatons in \fsm. This section describes the interface the visualization tool presents to users. The first subsection describes the basic visualization components without the interlacing of state invariants. The second subsection addresses how the users associate invariants with states and how the visualization tools displays invariant information.

\subsection{Visualizing Deterministic Finite-State Automata}

\begin{figure}[t]
\begin{center}
\includegraphics[scale=0.3]{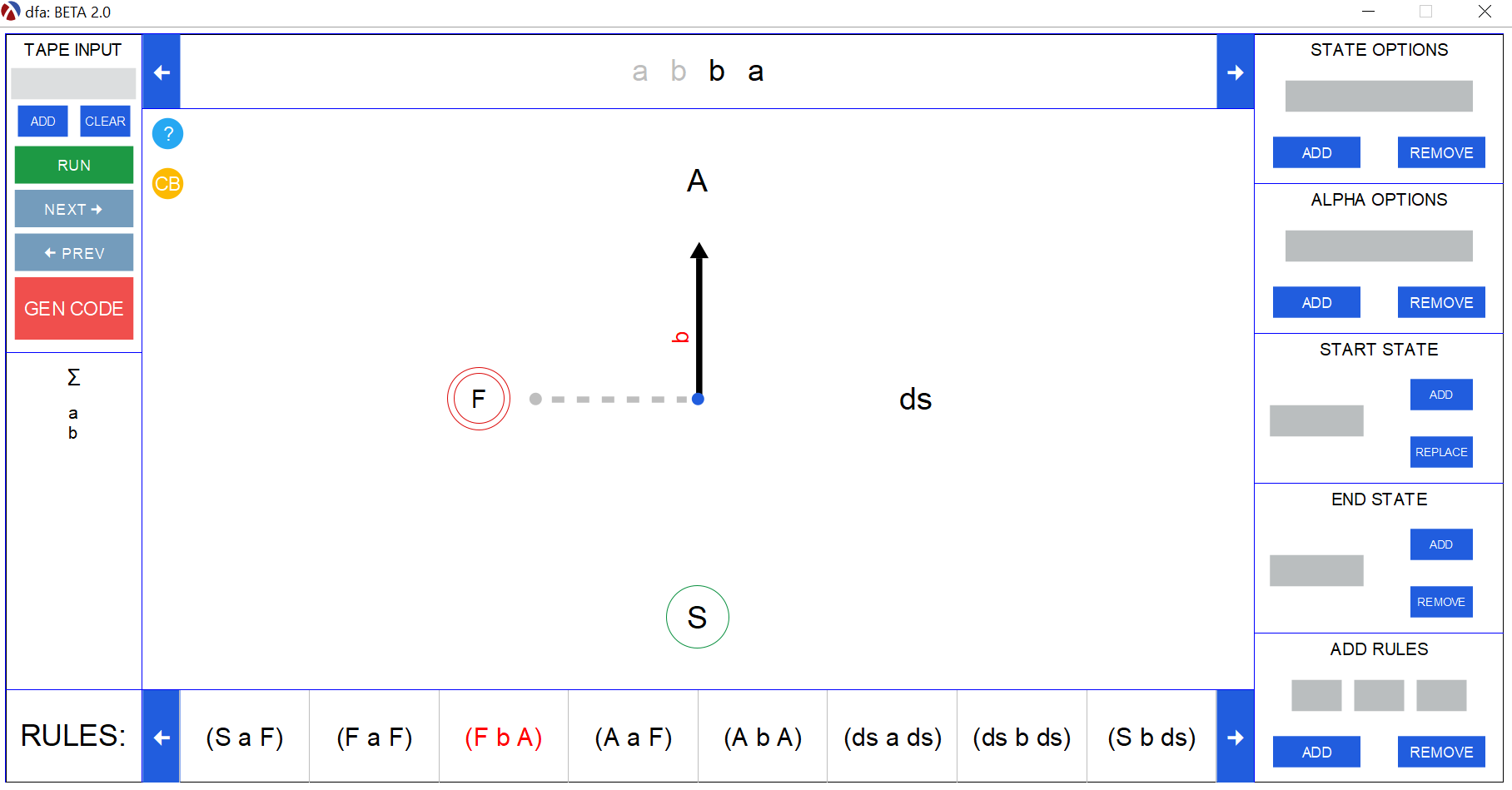}
\caption{The Basic Visualization Interface.} \label{tool1}
\end{center}
\end{figure}

The following principles guided the design of the visualization tool:
\begin{itemize}
  \item Integrate seamlessly into \fsm.
  \item Simulate and emphasize machine execution.
  \item Do not burden students with the creation of a graphical depiction.
  \item Allow for machine editing.
  \item Allow for invariant validation.
  \item Allow edited machines to be rendered as executable \fsm \ code.
\end{itemize}
In essence, the goal is to provide students with an interface that allows them to easily modify and interact with their machines at runtime.

\begin{figure}[t]
\begin{center}
\includegraphics[scale=0.3]{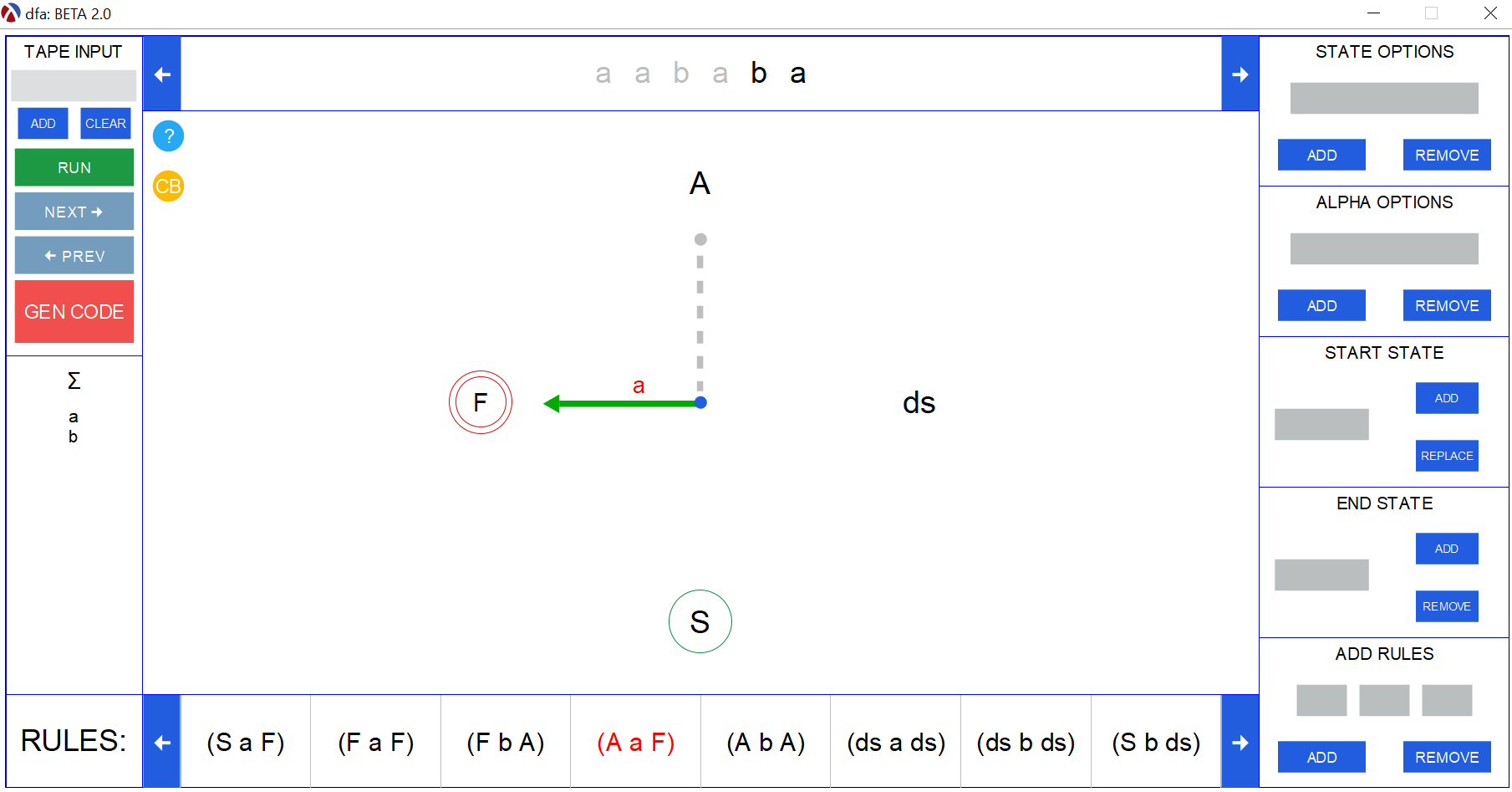}
\caption{Invariant Visualization During the Execution of \texttt{a*a}.} \label{tool2}
\end{center}
\end{figure}

Running the basic visualization tool is fully integrated into \fsm \ and is achieved by using one of two types of expressions:
\begin{description}
  \item[\texttt{(sm-visualize \textquotesingle{dfa})}] Launches the visualization tool and the user builds a \dfa \ from scratch.
  \item[\texttt{(sm-visualize $<$M$>$)}] Launches the visualization tool with machine \texttt{M} loaded.
\end{description}
Figure \ref{tool1} displays the basic interface of the \fsm \ visualization tool. The column on the right-hand side allows the user to add and remove states, alphabet elements, and rules. In addition, users can identify and replace the starting state and can add and remove final states.

The middle top displays the input tape. The consumed input is highlighted in a lighter color (e.g., \texttt{a b}). Below that, the machine's control representation is displayed. The states are organized in a circle, and a dark arrow indicates the current state. This arrow is labeled with the last input character consumed. The previous state is indicated with a light dashed line with a circle at its tail. The start state is enclosed in a single circle, and the final states are enclosed in two circles. In addition, there is a circled question mark that takes the user to the \fsm \ documentation page and a circled \texttt{CB} to switch the rendering to color blind mode. Below the control representation, the machine's transition rules are displayed. The rules are scrollable and the last rule used is highlighted in a lighter color.

The column on the left-hand side allows the user to clear or add to the input tape. The \texttt{RUN} button builds the machine in \fsm \ and applies it to the contents of the input tape. The \texttt{NEXT =>} and \texttt{<= PREV} buttons allow the user to sequentially move, respectively, forward and back through the sequence of configurations returned by \texttt{sm-showtransitions}. The \texttt{GEN CODE} button generates the \fsm \ code for the simulated machine and stores it in a file in the current directory. Every time this button is pressed, the definition of the current machine is added to this file with a date and time stamp. In this manner, students have access to a record of all the different versions of a machine for which they have chosen to generate code. Finally, the left-hand column also displays the machine's input alphabet.

The visualization tool, in this manner, allows students to focus on the algorithm they have developed based on state transitions and not on the graphical display of the machine. Students can, in essence, trace through a computation one step at a time to verify that the transitions make sense. In addition, students may also manually verify that the highlighted consumed input satisfies the invariant for the current state. Visually checking the validity of invariants, however, is prone to errors.

\subsection{Visualizing Invariants}
\label{inv}

Given the difficulty students face with visually checking invariants, the \fsm \ visualization tool permits users to associate an invariant with every state. The required syntax is:
\begin{alltt}
     (sm-visualize <\dfa> [(list <state> <predicate>)\(\sp{*}\)])
\end{alltt}
The optional list of state and predicate tuples associates a state with a predicate for the state's invariant. The predicate must take a single input representing the consumed input. The list of tuples does not need to be exhaustive. That is, not every state of the machine must be associated with a predicate. This allows students to build the invariant predicates in stages. Students can develop and test an invariant for a single state. Once they are satisfied that it is correct, they can then proceed with the invariant for another state.

\begin{figure}
\begin{alltt}
      (define a*a-buggy (make-dfa \textquotesingle(J K) \textquotesingle(a b)
                                  \textquotesingle{J} \textquotesingle(K)
                                  \textquotesingle((J a K) (K a K) (K b J))))

      (check-expect (sm-apply a*a-buggy \textquotesingle()) \textquotesingle{reject})
      (check-expect (sm-apply a*a-buggy \textquotesingle(b a a)) \textquotesingle{reject})
      (check-expect (sm-apply a*a-buggy \textquotesingle(a b a a)) \textquotesingle{accept})
      (check-expect (sm-apply a*a-buggy \textquotesingle(a a a)) \textquotesingle{accept})
      (check-expect (sm-apply a*a-buggy \textquotesingle(a b b a b a)) \textquotesingle{accept})

      (define (J-INV consumed-input)
        (or (empty? consumed-input) (not (eq? (last consumed-input) \textquotesingle{a}))))

      (check-expect (J-INV \textquotesingle()) #t)  (check-expect (J-INV \textquotesingle(a a)) #f)
      (check-expect (J-INV \textquotesingle(a b a b)) #t)

      (define (K-INV consumed-input)
        (and (eq? (first consumed-input) \textquotesingle{a}) (eq? (last consumed-input) \textquotesingle{a})))

      (check-expect (K-INV \textquotesingle()) #f) (check-expect (K-INV \textquotesingle(a)) #t)
      (check-expect (K-INV \textquotesingle(a b b)) #f) (check-expect (K-INV \textquotesingle(a a b a)) #t)

      (define (DS-INV consumed-input) (not (eq? (first consumed-input) \textquotesingle{a})))

      (check-expect (DS-INV \textquotesingle()) #f) (check-expect (DS-INV \textquotesingle(a a)) #f)
      (check-expect (DS-INV \textquotesingle(b a b)) #t)
\end{alltt}
\caption{A Buggy \dfa \ Implementation to Decide \texttt{a(a$|$b)$^*$a}.} \label{buggy1}
\end{figure}

Consider the visualization of \texttt{a*a} as follows:
\begin{alltt}
(sm-visualize
  a*a
  (list \textquotesingle{S} S-INV) (list \textquotesingle{F} F-INV) (list \textquotesingle{A} A-INV) (list DEAD DEAD-INV))
\end{alltt}
The invariants developed for \texttt{a*a} in Section \ref{design} are explicitly associated with their corresponding state. Figure \ref{tool2} captures a snapshot of the computation when \texttt{a*a} is applied to \textquotesingle{(a a b a b a)}. Observe that the arrow that points to the current state is green (a different tone in color blind mode). This indicates to the user that the invariant holds when \texttt{\textquotesingle{(a a b a)}} is consumed.

Each time the user clicks \texttt{NEXT =>} or \texttt{<= PREV}, the value of the invariant is recomputed. If the invariant does not hold, the arrow turns red. This indicates to the user that there is a flaw in their design or in their implementation. Students, therefore, now have their invariant predicates become part of their programs and no longer have to struggle through the cumbersome and error-prone manual validation of invariants as the machine executes. When the invariant fails, students now have a place to start their debugging process.

\section{Debugging Using the Visualization Tool}
\label{ex}

This section presents two examples of the use of the visualization tool in practice. The examples are submissions made by students. Both submissions have bugs, and this section shows how students use the visualization tool to detect and fix the bugs.

The first example is for a machine to decide the language that contains all strings that start and end with an \texttt{a} (i.e., the same language from Section \ref{design}). The second example is related to string matching, which is a central problem in computer systems and their applications \cite{Lewis}. In both cases, the visualization tool is used before proceeding with Step 8 of the design recipe for deterministic finite-state automata (i.e., before proving the machine's correctness).

\subsection{Debugging a Machine for \texttt{a(a$|$b)$^*$a}}
Consider the student submission for a machine to decide \texttt{a(a$|$b)$^*$a} displayed in Figure \ref{buggy1}. The student has designed a machine with three states and an invariant for each state. The state \texttt{J} represents that the input consumed is empty or the last character consumed is not an \texttt{a}. The only final state, \texttt{K}, represents that the first and last characters of the consumed input are \texttt{a}. Finally, the dead state represents that the consumed input is not empty and does not start with an \texttt{a}.

Upon testing the machine, all the unit tests for the invariants pass. The fifth unit test for \texttt{a*a-buggy}, however, fails. Any experienced instructor recognizes that the mistake made is typical for a novice: the transition from \texttt{J} to \texttt{ds} on a \texttt{b} is incorrect\footnote{Recall that for missing a transition in a \dfa \ \fsm \ adds a transition to the dead state.}. The student, however, is frustrated because some tests pass and (only) one fails. Random testing demonstrates to the student that the bug is common. For example, random testing yields the following results:
\begin{alltt}
     > (sm-test a*a-buggy 10)
     \textquotesingle{(((b a b a a b b a a) reject)
       ((b a a a a b) reject)
       ((a b b b b a a a a) reject)
       ((a a a a a) accept)
       ((a b b b a) reject)
       ((b a a) reject)
       ((a b a) accept)
       ((a b b a b) reject)
       ((b b b a a a a) reject)
       ((b b a a a) reject))}
\end{alltt}
The student is now convinced that the bug is unlikely to be obscure or to be a bug with \fsm, but is still unsure where the bug is.

\begin{figure}
\begin{center}
\includegraphics[scale=0.3]{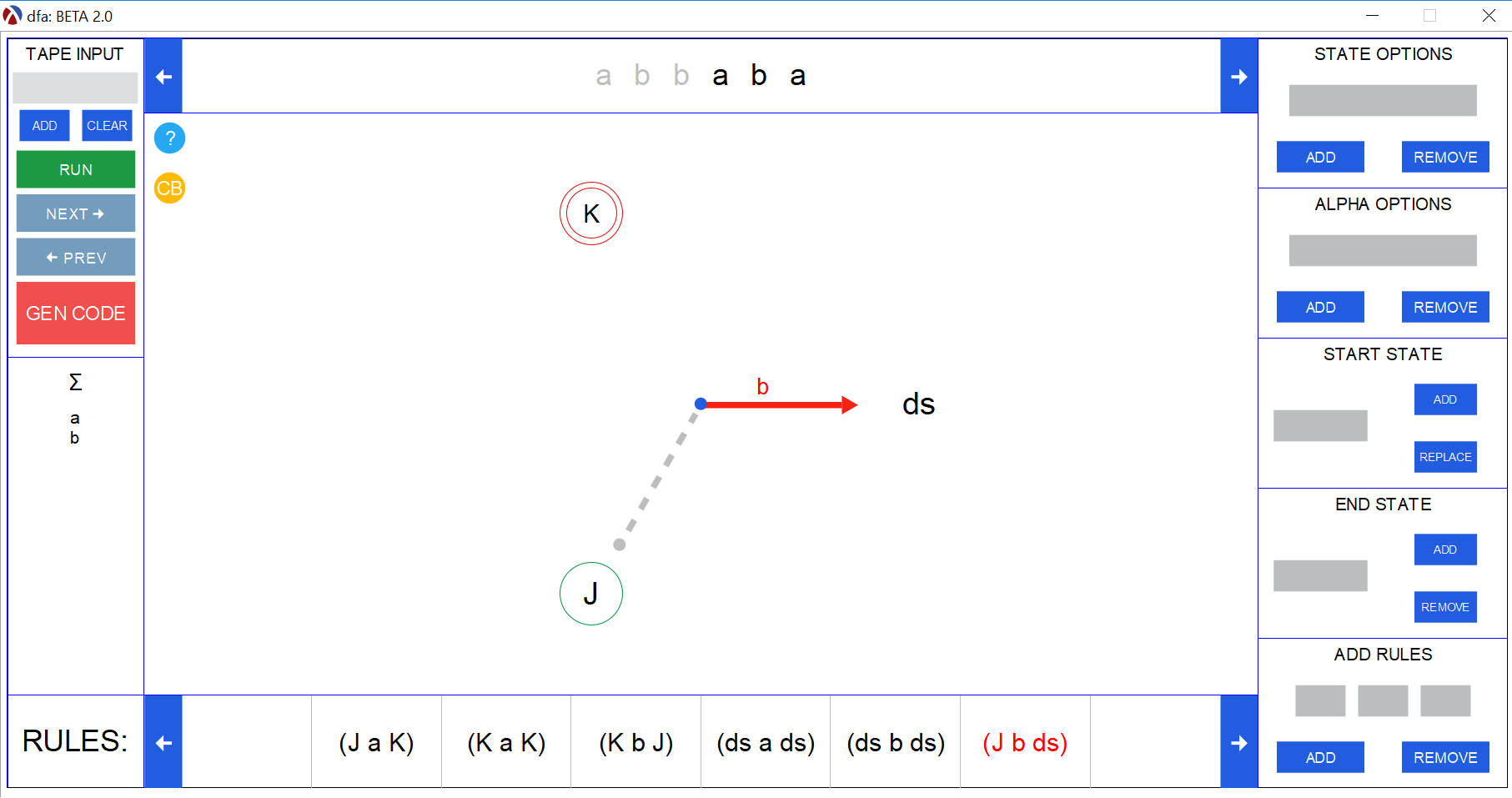}
\caption{Snapshot of the Execution of a Buggy Machine to Decide \texttt{a*a}.} \label{tool3}
\end{center}
\end{figure}

\begin{figure}
  \begin{alltt}
(define baba (make-dfa \textquotesingle{(A B C D F)} \textquotesingle{(a b)} \textquotesingle{A} \textquotesingle{(F)}
                       \textquotesingle{((A a A) (A b B) (B a C) (B b A) (F a F)
                         (C a A) (C b D) (D a F) (D b A) (F b F))}))
(check-expect (sm-apply baba \textquotesingle{(b a a a))} \textquotesingle{reject})
(check-expect (sm-apply baba \textquotesingle{(b a b b a b))} \textquotesingle{reject})
(check-expect (sm-apply baba \textquotesingle{(b a b a))} \textquotesingle{accept})
(check-expect (sm-apply baba \textquotesingle{(b b b a b a b a a))} \textquotesingle{accept})
(check-expect (sm-apply baba \textquotesingle{(a b b a b a b))} \textquotesingle{accept})

; A indicates that none of the pattern has been detected
(define (A-INV in)
  (cond [(empty? in) #t]
        [(= (length in) 1) (not (equal? (take-right in 1) \textquotesingle{(b)}))]
        [(= (length in) 2)
         (and (not (equal? (take-right in 1) \textquotesingle{(b)}))
              (not (equal? (take-right in 2) \textquotesingle{(b a)})))]
        [(= (length in) 3)
         (and (not (equal? (take-right in 1) \textquotesingle{(b)}))
              (not (equal? (take-right in 2) \textquotesingle{(b a)}))
              (not (equal? (take-right in 3) \textquotesingle{(b a b)})))]
        [else (and (not (equal? (take-right in 1) \textquotesingle{(b)}))
                   (not (equal? (take-right in 2) \textquotesingle{(b a)}))
                   (not (equal? (take-right in 3) \textquotesingle{(b a b)}))
                   (not (equal? (take-right in 4) \textquotesingle{(b a b a)})))]))
; B inidicates that b is detected
(define (B-INV in)
  (and (>= (length in) 1) (equal? (take-right in 1) \textquotesingle{(b)})))
; C indicates that ba is detected
(define (C-INV in)
  (and (>= (length in) 2) (equal? (take-right in 2) \textquotesingle{(b a)})))
; D indicates that bab is detected
(define (D-INV in)
  (and (>= (length in) 3) (equal? (take-right in 3) \textquotesingle{(b a b)})))
; F indicates that baba is detected
(define (F-INV in)
  (define (contains? patt w)
    (and (>= (length w) 4) (or (equal? (take w 4) patt)
         (contains? patt (rest w)))))
  (and (>= (length in) 4) (contains? \textquotesingle{(b a b a)} in)))
  \end{alltt}
\caption{A Proposed Solution to Decide if \texttt{baba} is a Substring of \texttt{w}.} \label{substr-buggy}
\end{figure}

At this point, the visualization tool is used to hone in on the bug. Figure \ref{tool3} displays a snapshot of the execution when the invariant fails for the first time using the word in the failed unit test. The machine has reached \texttt{ds} signifying that the consumed input does not start with an \texttt{a}. The red arrow in Figure \ref{tool3} clearly communicates that the invariant does not hold. At this point, the student realizes that they must refine the transition function. If the first input character read is \texttt{a}, then the machine should have no way to transition to the dead state. The student is asked if the invariants were used to design the transition function. The student confirms that they were not and understands that it is a mistake to skip design steps. Armed with this knowledge, the student goes back to designing and develops the machine presented in Section \ref{design}. Frustration over a bad grade and/or the lack of immediate feedback on their design and on their invariants is avoided. This exercise, in fact, has become a positive experience for both the student and the instructor.

\subsection{The String Matching Problem}
String matching algorithms may be designed using \dfas. There is a fixed string, \texttt{x}, that is called the \emph{pattern}. The goal is to decide if \texttt{x} is a substring of, \texttt{w}, a word. The algorithm is specifically designed around \texttt{x} and is not an algorithm that takes as input an arbitrary \texttt{x} and an arbitrary \texttt{w}.

As a class exercise, students are told that the input alphabet is \texttt{(a b)} and that \texttt{x} is \texttt{baba}. After attempting to use the design recipe, a student has developed the partial solution displayed in Figure \ref{substr-buggy}. In essence, the student's machine has a different state for each of the consecutive elements of the pattern detected. If the pattern is broken, the machine moves back to \texttt{A}, the starting state, indicating that none of the pattern has been detected. Although the student has formulated answers for the first six steps of the design recipe, it is a partial solution for two reasons. First, the student knows that the fifth unit test fails. Second, the student does not realize that the proposed state invariants are not strong enough.

\begin{figure}
\begin{center}
\includegraphics[scale=0.30]{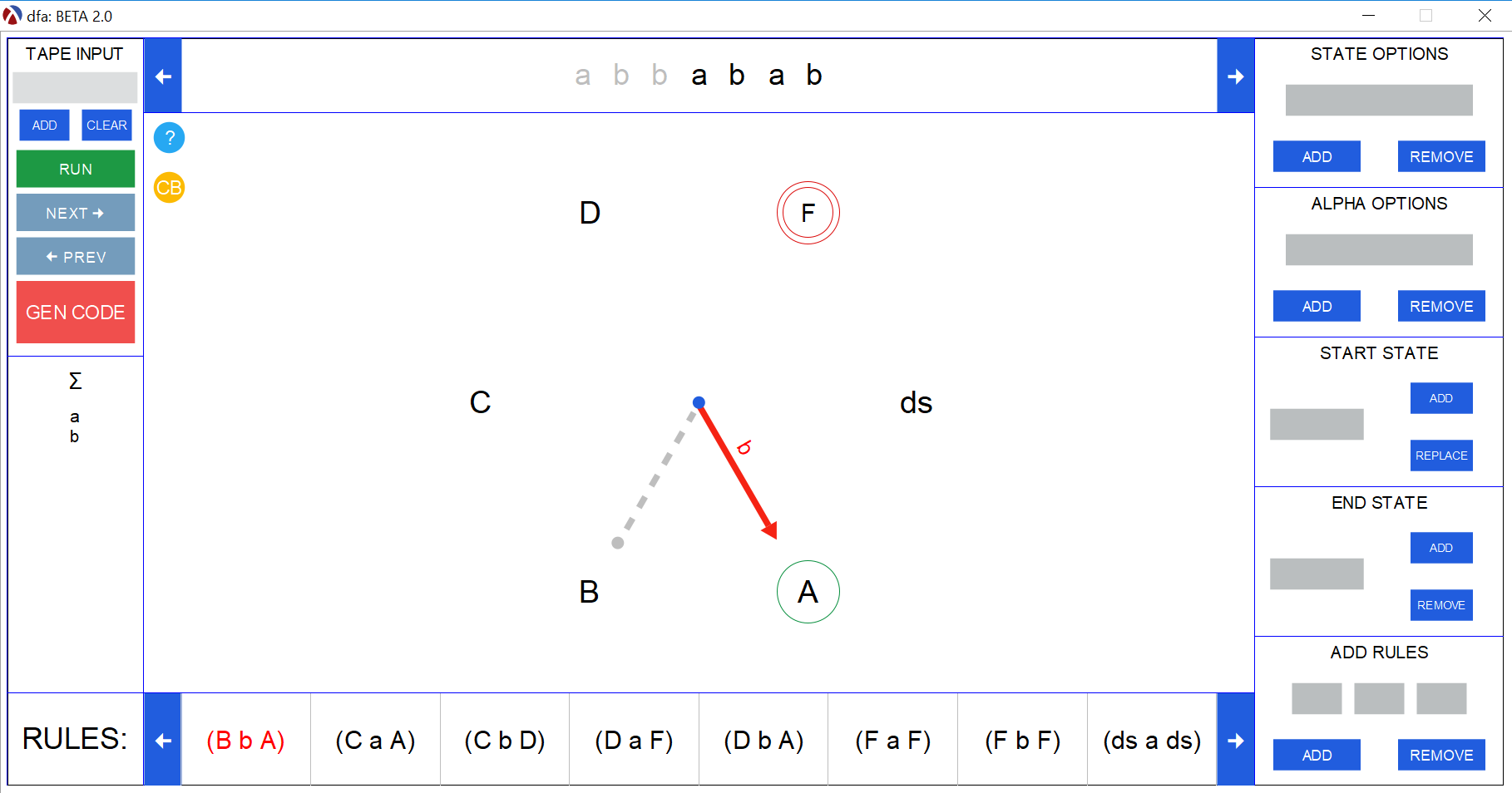}
\caption{Snapshot of the Execution of \texttt{baba} After Consuming the Second \texttt{b}.} \label{tool4}
\end{center}
\end{figure}

At this point, the student is asked to run the machine using \texttt{sm-visualize} and provide \texttt{abbabab}, the word in the failed test, as input to hone in on the bug. The snapshot of the visualization after consuming the second \texttt{b} is displayed in Figure \ref{tool4}. The student observes that the machine is in state \texttt{A}, indicating that none of the pattern has been detected. After some probing, the student realizes that reading a \texttt{b} in \texttt{B} does not mean that none of the pattern has been detected. The student becomes aware that the read \texttt{b} may be part of a different substring that matches the pattern and the machine ought to remain in state \texttt{B} and not move to state \texttt{A}.

\begin{figure}
\small{
\begin{alltt}
(define baba (make-dfa \textquotesingle{(A B C D F)} \textquotesingle{(a b)} \textquotesingle{A} \textquotesingle{(F)}
                       \textquotesingle{((A a A) (A b B) (B a C) (B b B) (F a F)
                         (C a A) (C b D) (D a F) (D b B) (F b F))}))
(define (contains? patt w)
  (and (>= (length w) (length patt))
       (or (equal? (take w (length patt)) patt) (contains? patt (rest w)))))
(define (A-INV input-consumed)
  (cond [(empty? input-consumed) #t]
        [(= (length input-consumed) 1)
         (not (equal? (take-right input-consumed 1) \textquotesingle{(b)}))]
        [(= (length input-consumed) 2)
         (and (not (equal? (take-right input-consumed 1) \textquotesingle{(b)}))
              (not (equal? (take-right input-consumed 2) \textquotesingle{(b a)})))]
        [(= (length input-consumed) 3)
         (and (not (equal? (take-right input-consumed 1) \textquotesingle{(b)}))
              (not (equal? (take-right input-consumed 2) \textquotesingle{(b a)}))
              (not (equal? (take-right input-consumed 3) \textquotesingle{(b a b)})))]
        [else (and (not (equal? (take-right input-consumed 1) \textquotesingle{(b)}))
                   (not (equal? (take-right input-consumed 2) \textquotesingle{(b a)}))
                   (not (equal? (take-right input-consumed 3) \textquotesingle{(b a b)}))
                   (not (contains? \textquotesingle{(b a b a)} input-consumed)))]))
(define (B-INV input-consumed)
  (and (>= (length input-consumed) 1)
       (equal? (take-right input-consumed 1) \textquotesingle{(b)})
       (cond [(= (length input-consumed) 2)
              (not (equal? (take-right input-consumed 2) \textquotesingle{(b a)}))]
             [(= (length input-consumed) 3)
              (and (not (equal? (take-right input-consumed 2) \textquotesingle{(b a)}))
                   (not (equal? (take-right input-consumed 3) \textquotesingle{(b a b)})))]
             [else (and (> (length input-consumed) 3)
                        (not (equal? (take-right input-consumed 2) \textquotesingle{(b a)}))
                        (not (equal? (take-right input-consumed 3) \textquotesingle{(b a b)}))
                        (not (contains? \textquotesingle{(b a b a)} input-consumed)))])))
(define (C-INV input-consumed)
  (and (>= (length input-consumed) 2) (equal? (take-right input-consumed 2) \textquotesingle{(b a)})
       (cond [(= (length input-consumed) 3)
              (not (equal? (take-right input-consumed 3) \textquotesingle{(b a b)}))]
             [else (and (> (length input-consumed) 3)
                        (not (equal? (take-right input-consumed 3) \textquotesingle{(b a b)}))
                        (not (contains? \textquotesingle{(b a b a)} input-consumed)))])))
(define (D-INV input-consumed)
  (and (>= (length input-consumed) 3)
       (equal? (take-right input-consumed 3) \textquotesingle{(b a b)})
       (if (> (length input-consumed) 3)
           (not (contains? \textquotesingle{(b a b a)} input-consumed))
           true)))
(define (F-INV input-consumed) (contains? \textquotesingle{(b a b a)} input-consumed))
\end{alltt} }
\caption{Changes to Proposed \dfa \ for Deciding if \texttt{baba} is a Substring of \texttt{w}.} \label{substr}
\end{figure}

After the above epiphany, the student is also advised to write unit tests for the invariant predicates. For example, the student is asked if state \texttt{D} only captures the idea that \texttt{bab} has been detected. Initially, the students responds affirmatively. The student is then asked if the consumed input contains the pattern upon the machine reaching state D. The answer is an emphatic no, and the student proceeds to explain why that is impossible. The argument is based on the entirety of the consumed input and not just the last three word elements. At this point, it is explained to the student that their argument is correct. However, the student is using a dynamic argument instead of a static argument. It is pointed out that if they know it is true, then it is invariant. Immediately, the student asks if it ought to be part of their invariant predicate for \texttt{D}. The student is told that it should be if it is useful in proving the partial correctness of the machine. This leads to a discussion about whether or not it is correct to reject a string that ends with the machine in state \texttt{D}. The student now realizes that \texttt{D} also represents that the pattern does not appear anywhere in the consumed input. The student is advised to think carefully about the invariants for the other states making sure they test the entire consumed input and to write unit tests for all predicates.

Figure \ref{substr} displays the student's refined solution\footnote{Unit tests for invariant predicates omitted in the interest of brevity.}. We can observe that the student has a deeper understanding of what each state represents. The invariants have been made stronger and the transition function of the machine reflects this understanding. For example, on a \texttt{b}, state \texttt{D} transitions to state \texttt{B} (i.e., \texttt{(D b B)} is part of the transition function). Furthermore, now \texttt{D} represents that \texttt{bab} has been detected, and the pattern does not appear in the consumed input. The effort made by the student has paid off because now the induction to prove the correctness of the machine is straightforward.

\subsection{Going Beyond the Design Recipe}
At this point, it is worthwhile to go beyond the design recipe with students. Students commonly observe that they found it useful to define \texttt{contains?} (displayed in Figure \ref{substr}), which decides if \texttt{patt} is present in \texttt{w}. Their natural question is why would they ever devise a \dfa \ if they can just write \texttt{contains?}. This is an excellent opportunity to reinforce lessons related to abstract running time. First, students are shown how easily \texttt{baba} is translated to standard \racket \ (non-\fsm) syntax:\\
\begin{alltt}
     (define (contains-baba? w)
       (define (dfa s w)
         (cond [(and (empty? w) (eq? s \textquotesingle{F})) #t]
               [(and (empty? w) (not (eq? s \textquotesingle{F}))) #f]
               [(eq? s \textquotesingle{A})
                 (if (eq? (first w) \textquotesingle{a})
                     (dfa \textquotesingle{A} (rest w))
                     (dfa \textquotesingle{B} (rest w)))]
               [(eq? s \textquotesingle{B})
                (if (eq? (first w) \textquotesingle{a})
                    (dfa \textquotesingle{C} (rest w))
                    (dfa \textquotesingle{B} (rest w)))]
               [(eq? s \textquotesingle{C})
                (if (eq? (first w) \textquotesingle{a})
                    (dfa \textquotesingle{A} (rest w))
                    (dfa \textquotesingle{D} (rest w)))]
               [(eq? s \textquotesingle{D})
                (if (eq? (first w) \textquotesingle{a})
                    (dfa \textquotesingle{F} (rest w))
                    (dfa \textquotesingle{B} (rest w)))]
               [else (dfa \textquotesingle{F} (rest w))]))
       (dfa \textquotesingle{A} w))
\end{alltt}
Now, students are asked what is the running time of \texttt{contains?} and of \texttt{contains-baba?}. After some discussion, students conclude that \texttt{contains?} is \texttt{O($n^2$)}, where \texttt{n} is the length of \texttt{w}. With some guidance, students realize that \texttt{contains-baba?} is \texttt{O($n$)}. Invariably, the question arises if it really matters (given that designing the \dfa \ is ``so much" work). An effective way to answer this is to run an experiment on a ``long" input that does not contain \texttt{baba} such as:
\begin{alltt}
     (define W (build-list 100000 (lambda (n) \textquotesingle{a})))

     (time (contains-baba? W))
     (time (contains? \textquotesingle{(b a b a)} W))
\end{alltt}
This yields the following results (after a significant amount of time waiting for \texttt{contains?} to terminate):
\begin{alltt}
     cpu time: 15 real time: 12 gc time: 0
     cpu time: 11406 real time: 11517 gc time: 0
\end{alltt}
The students usually chuckle and admit that it is worth the effort to design the \dfa. More importantly, they develop a true appreciation for \dfas \ and for exploring different design strategies beyond those found in a typical undergraduate algorithms textbook. Finally, a discussion is had about having \texttt{contains?} first create a \dfa \ for the pattern it receives as input and then perform the search, culminating in a discussion of the Knuth-Morris-Pratt algorithm \cite{Baase,KMP}.

\section{Concluding Remarks and Future Work}
\label{concl}
This article presents a novel design recipe for \dfas \ and \fsm's visualization tool to support the design process. The tool does not burden the user with having to render a \dfa \ as a graph. Instead, users define machines as executable code. The visualization tool allows the user to see the machine's state transitions during execution upon receiving input. Users may edit their machines and render their edited machines as executable code in \fsm. The most novel feature is that each state may be associated with an invariant predicate. The visualization tool indicates if the invariant holds after each transition. Both the code generation feature and the invariant validation feature are unique among software to visualize state-based machines. Providing a mechanism to visualize if a proposed invariant holds is useful for machine debugging and for formal machine verification as required by the presented design recipe.

Future work includes extending the visualization tool to encompass nondeterministic finite-state machines, pushdown automatas, and Turing machines. These extensions will follow the same interface presented to users for \dfas, which allows for the association of an invariant with each state and for the rendering of edited machines as executable \fsm \ code. In addition, we envision the visualization tool offering the user the option to display a \dfa \ using the current control representation or using a graph-based representation. Another interesting line of research that arises from this work is extending \fsm's testing facilities to include invariant testing. In addition to random testing, we envision \fsm \ generating a set of words that traverse every transition in a given machine and that reports where, if anywhere, an invariant fails. This ought to make \dfa \ testing more robust and facilitate the detection of bugs.

\section{Acknowledgements}
The authors thank the participants of \texttt{TFPIE 2020} for their insightful comments and feedback. In particular, we thank John Hughes for our discussions on testing after our presentation. It is thanks to these discussions that automatic exhaustive invariant testing is part of our future work. In addition, the authors thank Matthias Felleisen for his encouragement and his feedback on the work presented in this article. Both have proven useful in the development of the visualization tool. Finally, the authors also thank Seton Hall University for their support that has made this work possible.

\bibliographystyle{eptcs}
\bibliography{fsmvisualization}
\end{document}